\documentclass{article}
\usepackage[utf8]{inputenc}
\usepackage{placeins,hyperref}
\usepackage[margin=1.5in]{geometry}

\usepackage{graphicx}
\usepackage{xcolor}
\usepackage{tikz}

\newcommand\x{1.2}
\newcommand\y{2.5}
\newcommand\z{3.8}

\usepackage[normalem]{ulem}

\usepackage{amsmath}
\usepackage{amsthm}
\usepackage{amsfonts}
\theoremstyle{definition}
\newtheorem{defn}{Definition}
\newtheorem{theorem}{Theorem}
\newtheorem{prop}[theorem]{Proposition}
\newtheorem{lemma}[theorem]{Lemma}
\newtheorem{corollary}[theorem]{Corollary}

\usepackage{caption}
\usepackage{subcaption}

\usepackage[utf8]{inputenc}
\usepackage[T1]{fontenc}
\usepackage{authblk}

\usepackage{xcolor}

\definecolor{lime}{HTML}{A6CE39}
\DeclareRobustCommand{\orcidicon}{%
	\begin{tikzpicture}
	\draw[lime, fill=lime] (0,0) 
	circle [radius=0.16] 
	node[white] {{\fontfamily{qag}\selectfont \tiny ID}};
	\draw[white, fill=white] (-0.0625,0.095) 
	circle [radius=0.007];
	\end{tikzpicture}
	\hspace{-2mm}
}

\foreach \x in {A, ..., Z}{%
	\expandafter\xdef\csname orcid\x\endcsname{\noexpand\href{https://orcid.org/\csname orcidauthor\x\endcsname}{\noexpand\orcidicon}}
}

\title{The effect of the Katz parameter on node ranking, with a medical application}
\author[1,*]{Hunter~Rehm\orcidA{}}
\author[2]{Mona~Matar\orcidB{}}
\author[1]{Puck~Rombach\orcidC{}}
\author[2]{Lauren~McIntyre}

\affil[1]{Department of Mathematics and Statistics, University of Vermont, Burlington, VT, USA}
\affil[2]{NASA Glenn Research Center, Cleveland, OH, USA}

\affil[*]{Corresponding author: Hunter.Rehm@uvm.edu}

\begin{document}

\maketitle

\begin{abstract}
Katz centrality is a popular network centrality measure. It takes a (weighted) count of all walks starting at each node, with an additional damping factor of $\alpha$ that tunes the influence of walks as lengths increase. We introduce a tool to compare different centrality measures in terms of their node rankings, which takes into account that a relative ranking of two nodes by a centrality measure is unreliable if their scores are within a margin of error of one another. We employ this tool to understand the effect of the $\alpha$-parameter on the lengths of walks that significantly affect the ranking of nodes. In particular, we find an upper bound on the lengths of the walks that determine the node ranking up to this margin of error. If an application imposes a realistic bound on possible walk lengths, this set of tools may be helpful to determine a suitable value for $\alpha$. We show the effect of $\alpha$ on rankings when applied to the Susceptibility Inference Network, which contains subject matter expert informed data that represents the probabilities of medical conditions progressing from one to another. This network is part of the Medical Extensible Dynamic Probabilistic Risk Assessment Tool, developed by NASA, an event-based risk modeling tool that assesses human health and medical risk during space exploration missions.\\



\end{abstract}

\section{Introduction}
Networks are structures that naturally appear in every aspect of life and are studied in a wide range of disciplines from sociology, to biology, and engineering~\cite{freeman2004development,guze2014graph,newman2018networks,pavlopoulos2011using}. One common question asked in network theory is how to rank the nodes according to importance, where importance can have many meanings depending on the application~\cite{das2018study}. Many ranking algorithms are based on a, possibly weighted, count of walks that a node is contained in. Examples of such centrality measures are degree, betweenness~\cite{freeman1977set}, closeness~\cite{bavelas1950communication,beauchamp1965improved,freeman1978centrality}, eigenvector~\cite{bonacich1972technique,bonacich1987power}, PageRank~\cite{page1999pagerank}, and Katz centrality~\cite{katz1953new}. The latter is the focus of this paper. 

\textit{Katz centrality}, developed by Leo Katz in 1953~\cite{katz1953new}, has been used in numerous applications~\cite{fletcher2018structure,zhan2017identification}. The Katz centrality of a node is a weighted count of all walks of any length starting at the node. Each walk of length $k$ is weighted by $\alpha^k$, where $\alpha$ is called the \textit{Katz parameter}. We give a formal definition of Katz centrality in Section~\ref{sec:Definitions}. 
Since the Katz parameter has a decaying effect, we can approximate the Katz centrality by ignoring the contribution of walks past a given length $L$. In~\cite{nathan2017approximating} and \cite{nathan2017graph}, the authors numerically explore this type of approximation. In Section~\ref{sec:ParameterLengthRelationship}, we give a lower bound on this value $L$ (in terms of $\alpha$) that guarantees a desired level of accuracy in terms of the Katz centrality and its node ranking. 

As an application, we look at a model developed by the National Aeronautics and Space Administration (NASA) called the Susceptibility Inference Network (SIN). This network models a set of medical conditions that may progress from one to another with some probability. The application of Katz centrality in this setting estimates the effect of a medical condition and its consequent progressions on astronaut productivity. NASA wishes to use the progression risk to identify a condition (or group of conditions) that highly influences the risk of the crew members.

This paper is organized as follows. Section~\ref{sec:Definitions} reviews useful graph theory concepts, definitions and basic results. Section~\ref{sec:SIN_intro} introduces a network of medical conditions created by subject matter experts, which we will later use as a case study. In Section~\ref{sec:ParameterLengthRelationship}, we bound the error generated from approximating the Katz centrality by restricting the number of steps allowed in a walk, and develop a relationship between that number and the Katz parameter $\alpha$. An example of the relationship between $\alpha$ and the $\alpha$-Katz centrality node ranking is in Section \ref{sec:ParameterLengthRelationship} and our medical application in Section~\ref{sec:ApplicationToSIN}. Additionally, we assess the upper bound given in Section \ref{sec:ParameterLengthRelationship} to the true length in Section \ref{sec:comparisons}. Finally, the results and future work are addressed in Section~\ref{sec:Discussion}.

\subsection{Definitions}\label{sec:Definitions}
In this section, we provide basic definitions for the graph theoretical structures and tools used for the results in Section~\ref{sec:ParameterLengthRelationship}.

\begin{defn}\label{def:GAW} {\bf (Weighted, directed network)}
Let $N = (V,E,w)$ be an \textit{edge-weighted, directed network} consisting of $V$, the set of \textit{n} \textit{nodes}, $E\subseteq V\times V$, the set of \textit{edges}, and a weight function $w:E \to \mathbb{R}^+$. 

We represent such a network by an \textit{adjacency matrix} $A=A(N)$, where the entry $A_{ij}$ is the weight of the edge from node $i$ to node $j$, or $A_{ij} = 0$ if there is no edge from $i$ to $j$ in $N$. Let $W$ be an $n$-dimensional vector of non-negative node weights. In a setting where edges and/or nodes are unweighted, weights in $A$ and $W$ are all set to 1. 
\end{defn}

For example, in our application in Section \ref{sec:ApplicationToSIN}, our weighted, directed network has nodes that represent medical conditions and the node weights represent their severity, while edge weights represent the probability of one medical condition progressing to another.

The \textit{spectral radius} $\rho$ of $N$ (or $A$) is the maximum modulus of the eigenvalues of $A$. A \textit{walk of length $k$} from node $u$ to $v$ is a sequence of $k$ edges $(v_i,v_{i+1})\in E$, $i\in [1,k]$ such that $v_1 = u$ and $v_{k+1} = w$.
The \textit{distance} from $u$ to $v$ is the length of a shortest walk from $u$ to $v$. The \textit{$k$-hop neighborhood} of a node $v\in V$ is the set of nodes at a distance less than or equal to $k$ from $v$. A \textit{centrality measure} is a function that assigns a real number to each node, in order to evaluate its relative importance to other nodes. Each centrality measure gives a (partial) \textit{ranking} of the nodes, which reflects their relative importance. 

\begin{defn}\label{def:Katz} {\bf (Katz centrality)}
Let $N$ be an edge-weighted, directed network with node weights $W$. Let $A=A(N)$ with spectral radius $\rho$, and let $\alpha\in \left(0,1/\rho\right)$. The \textit{$\alpha$-Katz centrality} vector~\cite{de2019analysis,estrada2010network,katz1953new}, is defined as 
$$C(\alpha) = \left(\sum_{i = 1}^\infty \alpha^kA^k\right)\cdot W = \left(\left(I-\alpha A\right)^{-1}-I\right)\cdot W.$$
The \textit{$\alpha$-Katz score} of a particular node $i$ can then be expressed as 
$$C(\alpha)_i = \sum_{k = 1}^\infty \sum_{j = 1}^n W_j\alpha^k \left(A^k\right)_{ij}.$$
%quantifies the role of node $i$ in sending information through the network. 
The \textit{$(\alpha, \ell)$-Katz centrality} vector~\cite{acar2009link,beres2018temporal,lu2010supervised} is
$$C(\alpha,\ell) = \left(\sum_{i = 1}^\ell \alpha^kA^k\right)\cdot W$$
and for a particular node $i$, the \textit{$(\alpha, \ell)$-Katz score} can be written as
$$C(\alpha,\ell)_i = \sum_{k = 1}^\ell \sum_{j = 1}^n W_j\alpha^k \left(A^k\right)_{ij}.$$
\end{defn}

Both $C(\alpha)$ and $C(\alpha, \ell)$ measure the \textit{downstream} influence of nodes, since they are weighted sums over outgoing walks. Replacing the matrix $A$ by its transpose $A^T$ reverses edge directions, which results in taking weighted sums over incoming walks instead, and measuring the \textit{upstream} influence~\cite{de2019analysis,newman2018networks}.

Definition~\ref{def:agree} provides a tool to compare centrality measures purely in terms of their relative node rankings. Intuitively, we may set a threshold $\epsilon$ for a centrality measure $C$, such that $\lvert C_i - C_j\rvert\geq \epsilon$ implies that $C$ provides a relative ranking of nodes $i$ and $j$. If $\lvert C_i - C_j\rvert < \epsilon$, we cannot reliably recover a ranking from $C$. For two centrality measures $C$ and $C'$, we compare their rankings and conclude that they agree on a ranking if they agree for every node pair where both of their rankings are reliable. 

\begin{defn}\label{def:agree} {\bf ($\epsilon$-agreement)}
Let $N$ be a weighted, directed network, $\epsilon,\epsilon'>0$, and $C$ and $C'$ be centrality measures. The nodes $i,j\in V(N)$ are \textit{$(\epsilon,\epsilon')$-properly ranked with respect to $C$ and $C'$} if the following holds:
\begin{enumerate}
    \item $\lvert C_i - C_j\rvert<\epsilon \text{ or } \lvert C'_i - C'_j\rvert<\epsilon',$ 
    \item otherwise, $C_i-C_j$ and $C'_i-C'_j$ have the same sign.
\end{enumerate}
We say that $C$ and $C'$ \textit{$(\epsilon,\epsilon')$-agree on $N$} if every pair of nodes in $N$ is $(\epsilon,\epsilon')$-properly ranked with respect to $C$ and $C'$. If $\epsilon=\epsilon'$, we simply say $\epsilon$-proper ranking and $\epsilon$-agreement.
\end{defn}

\begin{prop}\label{prop:epsilonagree}
Let $C$ and $C'$ be two centrality measures on a network $N$. If \[\|C - C'\|_\infty <\epsilon,\] then $C$ and $C'$ $\epsilon$-agree.
\end{prop}

\begin{proof}
Suppose for the sake of contradiction that $C$ and $C'$ do not $\epsilon$-agree. Then, there exists a pair of nodes $u,v\in V(N)$ that is not $\epsilon$-properly ranked. By part (1) of Definition~\ref{def:agree}, we have $\lvert C_u - C_v\lvert>\epsilon$ and $\lvert C_u' - C_v'\lvert>\epsilon$. By part (2), without loss of generality, we have $C_u-C_v<0$ and $C_u'-C_v'>0$.
 
We have
\begin{align*}
    C_u'-C_u&=    \underbrace{C_u'-C_v'}_{>\epsilon}+\underbrace{C_v'-C_v}_{>-\epsilon}+\underbrace{C_v-C_u}_{>\epsilon}> \epsilon,
\end{align*}
which contradicts that $\|C-C'\|_\infty<\epsilon$.
\end{proof}

\begin{prop}\label{prop:epsilonagree2}
Let $C$ and $C'$ be two centrality measures on a network $N$ such that for all $v\in V(N)$, $0\leq C_v - C'_v\leq 2\epsilon $, then $C$ and $C'$ $\epsilon$-agree.
\end{prop}

\begin{proof}
Suppose for the sake of contradiction that $C$ and $C'$ do not $\epsilon$-agree. Then, there exists a pair of nodes $u,v\in V(N)$ that is not $\epsilon$-properly ranked. By part (1) of Definition~\ref{def:agree}, we have $\lvert C_u - C_v\lvert>\epsilon$ and $\lvert C_u' - C_v'\lvert>\epsilon$. By part (2), without loss of generality, we have $C_u-C_v<0$ and $C_u'-C_v'>0$.
 
 We have
 \begin{align*}
     C_u-C_u'&=    \underbrace{C_u-C_v}_{<-\epsilon}+\underbrace{C_v-C_v'}_{<2\epsilon}+\underbrace{C_v'-C_u'}_{<-\epsilon}< 0,
 \end{align*}
 which contradicts that $C_u - C_u'\geq 0$.
\end{proof}

In Section~\ref{sec:ParameterLengthRelationship}, we use this notion to compare two closely related centrality measures, $C(\alpha)$ and $C(\alpha, \ell)$, and we therefore only use $\epsilon$-proper ranking and $\epsilon$-agreement. However, we state the definition here in a more general form, as it can be used to compare any pair of centrality measures, even if their distributions of values differ significantly. The vector $C(\alpha,\ell)$ converges to $C(\alpha)$ as $\ell \to \infty$. In Theorem~\ref{thm:epsilonagree}, we show that this implies that for all $\epsilon>0$, there exists an $L$ so that for any $\ell>L$, $C(\alpha)$ and $C(\alpha,\ell)$ $\epsilon$-agree.

\subsection{Susceptibility Inference Network}\label{sec:SIN_intro}
The Medical Extensible Dynamic Probabilistic Risk Assessment Tool (MEDPRAT) developed by the National Aeronautics and Space Administration (NASA) is an event-based risk modeling tool that assesses human health and medical risk during space exploration missions~\cite{mcintyre2020dynamic,mcintyre2022modelbasedrisk}. One of its key tasks is to capture relationships between medical events. These relationships are described by the Susceptibility Inference Network (SIN).

The SIN is a directed network where nodes represent medical conditions. The data in this network is subject matter expert informed and is currently a prototype. There is an edge from $u$ to $v$ if medical condition $u$ can progress into medical condition $v$. This directed edge $(u,v)$ is weighted by the probability that such a progression occurs. Note that a medical conditions may progress to multiple other conditions simultaneously, or to no other conditions. Therefore, this matrix is not a transition matrix. The SIN currently has 99 nodes and 1078 edges. Medical conditions included are, for example, acute radiation syndrome, which has many outgoing edges towards other medical conditions. On the other hand, anxiety has many incoming edges.

Each node in the SIN has an associated weight which evaluates the severity of having the condition regardless of the progression from or to that condition. This severity is quantified by Quality Time Lost (QTL), which we call \textit{primary QTL}. The primary QTL measures the productive time that a crew member is expected to lose, measured in days. 

For each medical condition, we wish to estimate the total expected QTL resulting from the condition itself as well as its possible progressions. We let $W$ represent the nodes' expected primary QTL, and the total expected QTL for a medical condition $i$ is given by
\[ \sum_{k = 0}^\infty \sum_{j = 1}^n W_j \left(A^k\right)_{ij} = W_i + C(1)_i. \]
\[ C(1)_i = \sum_{k = 1}^\infty \sum_{j = 1}^n W_j \left(A^k\right)_{ij} .\]
The value $ C(1)_i$ can be viewed as the expected \textit{subsequent QTL} of a medical condition, caused by progressions and assuming the absence of interventions and time limits of the mission. Then, to make this estimate more realistic, the parameter $\alpha$ provides a damping factor which decreases the weight of walks as they get longer. This application therefore illustrates the importance of using realistic walk lengths to guide the choice of $\alpha$. We provide a theoretic foundation for this in Section~\ref{sec:ParameterLengthRelationship}. In Section~\ref{sec:ApplicationToSIN}, we discuss how different values of $\alpha$ produce different rankings for the SIN due to subsequent QTL.

\section{The Katz parameter and walk length} \label{sec:ParameterLengthRelationship}

In this section, we describe in more detail the relationship between maximum walk lengths $\ell$ and the parameter $\alpha$. In Theorem~\ref{thm:epsilonagree}, we find a lower bound on $\ell$ that guarantees that $C(\alpha)$ and $C(\alpha, \ell)$ $\epsilon$-agree. This sheds light on the length of walks that decide the ranking provided by $C(\alpha)$. We provide a small, illustrative example of the effect of $\alpha$ on node rankings.

Lemma~\ref{lem:errorbound} gives an upper bound on the difference between values in $C(\alpha)$ and $C(\alpha,\ell)$. 

\begin{lemma}(Absolute Error Tolerance)\label{lem:errorbound}
Let $p \in \{ 1,2,\infty \}$ and $\alpha \in (0,1/\rho)$. Then 
$$\left\|C(\alpha)-C(\alpha,\ell)\right\|_\infty\le\left(\alpha\|A\|_p\right)^\ell\left\|C(\alpha)\right\|_p\\:=\epsilon_\ell.$$
\end{lemma}

\begin{proof}
First, note that $\|V\|_\infty \leq \|V\|_2 \leq \|V\|_1$ for all vectors $V$. Furthermore, the $p$-norm is submultiplicative. 
We have

\begin{align*}
    \|C(\alpha) - C(\alpha,\ell)\|_\infty &\leq \|C(\alpha) - C(\alpha, \ell)\|_p\\
    &= \left\|\left(\sum_{i = \ell+1}^\infty \alpha^k A^k\right)\cdot W\right\|_p\\
    &=  \left\|\alpha^{\ell}A^{\ell}\left(\sum_{i = 1}^\infty \alpha^k A^k\right)\cdot W\right\|_p\\
    &\leq \alpha^{\ell}\|A^{\ell}\|_p\left\|\left(\sum_{i = 1}^\infty \alpha^k A^k\right)\cdot W\right\|_p\\
    &\leq \left(\alpha\|A\|_p\right)^\ell\left\|C(\alpha)\right\|_p.
\end{align*}
\end{proof}

In Lemma~\ref{lem:error_guarantee}, we show that there exists an $L$ so that for all $\ell>L$, the difference between the scores in $C(\alpha)$ and $C(\alpha,\ell)$ is small. 

\begin{lemma}(Error Tolerance Guarantee)\label{lem:error_guarantee}
Let $p \in \{ 1,2,\infty \}$, $\alpha\in (0,1/\|A\|_p)$ and $\epsilon>0$. If
$$\ell > \log_{\alpha \|A\|_p}\left(\frac{2\epsilon}{\|C(\alpha)\|_p}\right):=L$$
then 
$\left\|C(\alpha)-C(\alpha,\ell)\right\|_\infty<2\epsilon$.
\end{lemma}

\begin{proof}
Note that $\alpha <1/\|A\|_p\leq 1/\rho$. By Lemma~\ref{lem:errorbound}, it suffices to show that $\epsilon_\ell< 2\epsilon$ when $\ell>L$. We have

\begin{align*}
    \epsilon_{\ell} &= \left\|C(\alpha)\right\|_p\left(\alpha\|A\|_p\right)^\ell\\
    &< \left\|C(\alpha)\right\|_p\left(\alpha\|A\|_p\right)^{L}\\
    &= \left\|C(\alpha)\right\|_p\left(\alpha\|A\|_p\right)^{\log_{\alpha \|A\|_p}\left(\frac{2\epsilon}{\|C(\alpha)\|_p}\right)}\\
    &=\|C(\alpha)\|_p\frac{2\epsilon}{\|C(\alpha)\|_p}\\
    &= 2\epsilon.
\end{align*}
\end{proof}

Lemma~\ref{lem:error_guarantee} bounds the difference between $C(\alpha)$ and $C(\alpha,\ell)$ when $\ell$ is large enough. Theorem~\ref{thm:epsilonagree} shows that this also ensures that $C(\alpha)$ and $C(\alpha,\ell)$ $\epsilon$-agree.

\begin{theorem}\label{thm:epsilonagree}
Let $p \in \{ 1,2,\infty \}$, $\alpha\in (0,1/\|A\|_p)$, $\epsilon>0$ and $L$ be as in Lemma~\ref{lem:error_guarantee}. If $\ell > L$ then $C(\alpha)$ and $C(\alpha,\ell)$ $\epsilon$-agree.
\end{theorem}

\begin{proof} By Lemma~\ref{lem:error_guarantee} we have that $\left\|C(\alpha)-C(\alpha,\ell)\right\|_\infty<2\epsilon$. Therefore, by Proposition~\ref{prop:epsilonagree2}, we have that $C(\alpha)$ and $C(\alpha,\ell)$ $\epsilon$-agree.
\end{proof}
We choose $\epsilon$ so that two nodes that have Katz scores within $\epsilon$ of each other can be considered equivalent in terms of ranking. Of course, such a value must be chosen relative to the Katz scores. In Corollary~\ref{cor:relativeerror}, we suggest letting $\epsilon$ be some fraction of $\|C(\alpha)\|_p$, and $\alpha$ a fraction of $1/\|A\|_p$ 
for $p\in \{1,2,\infty\}$. 

\begin{corollary}[Relative Error]\label{cor:relativeerror}
Let $p \in \{ 1,2,\infty \}$ and $\alpha_0,\epsilon_0\in (0,1)$. For $\alpha = \alpha_0/\|A\|_p$, and $\epsilon = \epsilon_0\|C(\alpha)\|_p$, if $\ell>\log_{\alpha_0}(\epsilon_0) = L$ then $C(\alpha)$ and $C(\alpha,\ell)$ $\epsilon$-agree.
\end{corollary}

We present a small example to illustrate the effect of $\alpha$ on node rankings. Consider the edge-weighted directed network $N$ in Figure \ref{fig:smallkatznetwork}. 
Let $\alpha_1$ be the positive solution to $C(b)=C(c)$, let $\alpha_2$ be the positive solution to $C(a)=C(c)$, and let $\alpha_3$ be the positive solution to $C(a)=C(b)$. When $\alpha$ is small, the walks of length 1 determine the ranking. As $\alpha$ increases, walks of length 2 and later walks of length 3 become more important. node $c$ has the most walks of length 1, node $b$ the most of length 2, and node $a$ the most of length 3, and they are each ranked on top for different ranges of $\alpha$. The change in rankings as function of $\alpha$ is sketched in Figure~\ref{fig:smallkatzexample}.

\begin{figure}[ht]
\begin{subfigure}[t]{.35\textwidth}\centering
      \includegraphics[scale = .65]{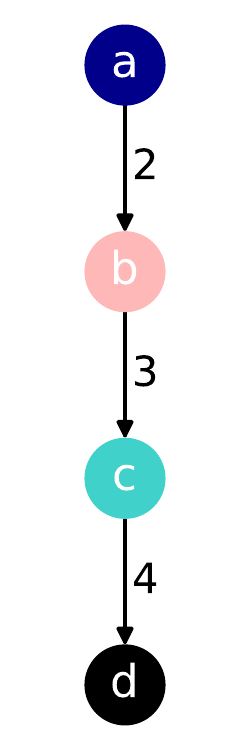}
    \caption{An example of a weighted directed network.}
    \label{fig:smallkatznetwork}
\end{subfigure}%
\hfill
\begin{subfigure}[t]{.6\textwidth}\centering
    \includegraphics[scale = .4]{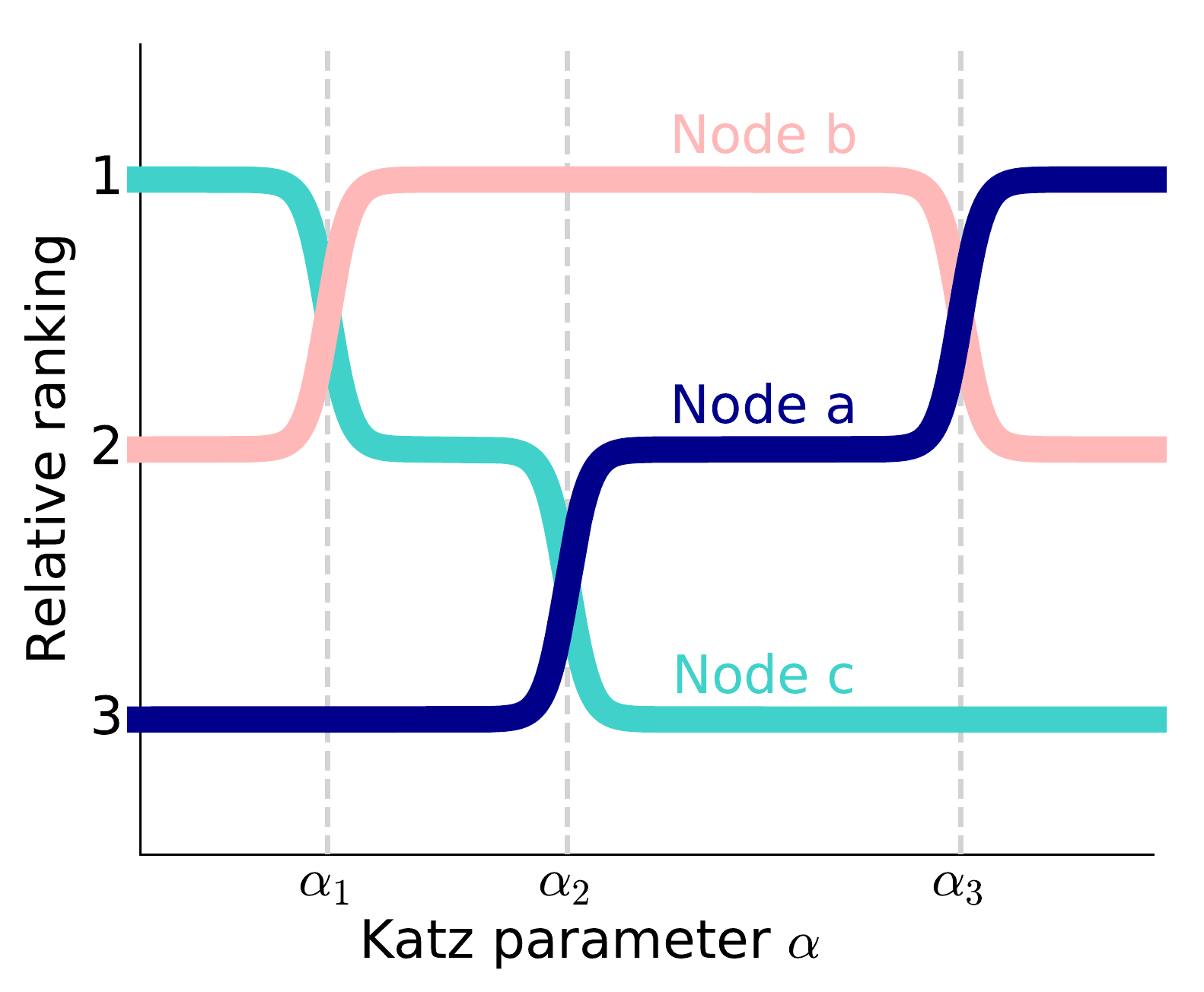}
    \caption{The node rankings from the Katz scores of the network in Figure \ref{fig:smallkatznetwork}.}
    \label{fig:smallkatzexample}
\end{subfigure}
\caption{The relationship between the Katz parameter $\alpha$ and the vertex ranking in a directed, edge-weighted example network.}
\end{figure}

\section{Application to the Susceptibility Inference Network}\label{sec:ApplicationToSIN}

We use $\alpha$-Katz centrality to rank medical conditions in the Susceptibility Inference Network (SIN) by their expected subsequent QTL. We create the distribution of difference between Katz scores and set $\epsilon$ so that $5$\% of the differences are less than it. This implies that 95\% of the pairs are $\epsilon$-properly ranked. We run the analysis for $\alpha=1$ and $\alpha=0.4$.  When $\alpha=1$, $(\alpha,\ell)$-Katz centrality and $\alpha$-Katz centrality $\epsilon$-agree when $\ell\geq 5$, and when $\alpha=0.4$, $(\alpha,\ell)$-Katz centrality and $\alpha$-Katz centrality $\epsilon$-agree when $\ell\geq 3$. Table~\ref{tab:SIN_ranking} shows a comparison of rankings of the top 10 conditions, and Figure \ref{fig:SIN_subgraph} illustrates the subnetwork of the SIN containing the 12 nodes that appear in Table~\ref{tab:SIN_ranking}. 

We note that one significant change in the ranking as $\alpha$ decreases from 1 to $0.4$ is the one that swaps the positions of nodes $I$ and $J$. Figure~\ref{fig:IJcontributions} plots the contributions of walks of length $\ell$ to the Katz scores of these two nodes. When $\alpha=1$, the contribution of walks of lengths 2 and 3 contribute enough to the score of node $I$ to place it above $J$. When $\alpha=0.4$ these longer walks contribute significantly less, lowering the ranking of node $I$ to fall below that of $J$. In this and other applications, there is no universal value of $\alpha$ that is best. Here, one would have to take into account the length of the space mission, as well as other factors that affect how realistic any number of progressions is. 

\begin{table}[ht]
    \centering
    \begin{tabular}{l|c|c}
      Rank & $\alpha = 1$ & $\alpha = 0.4$  \\
      \hline
      1 &  A &  A\\
      2 &  B &  B\\
      3 &  C &  D\\
      4 &  D &  C\\
      5 &  E &  E\\
      6 &  F &  F\\
      7 &  G &  G\\
      8 &  H &  J\\
      9 &  I &  H\\
      10 &  J &  L
    \end{tabular}
    \caption{The effect of the choice of $\alpha$ on the ranking of medical conditions in the SIN due to subsequent QTL.}
    \label{tab:SIN_ranking}
\end{table}

\begin{figure}[ht]
     \centering
     \begin{subfigure}[b]{0.45\textwidth}
         \centering
         \includegraphics[width=\textwidth]{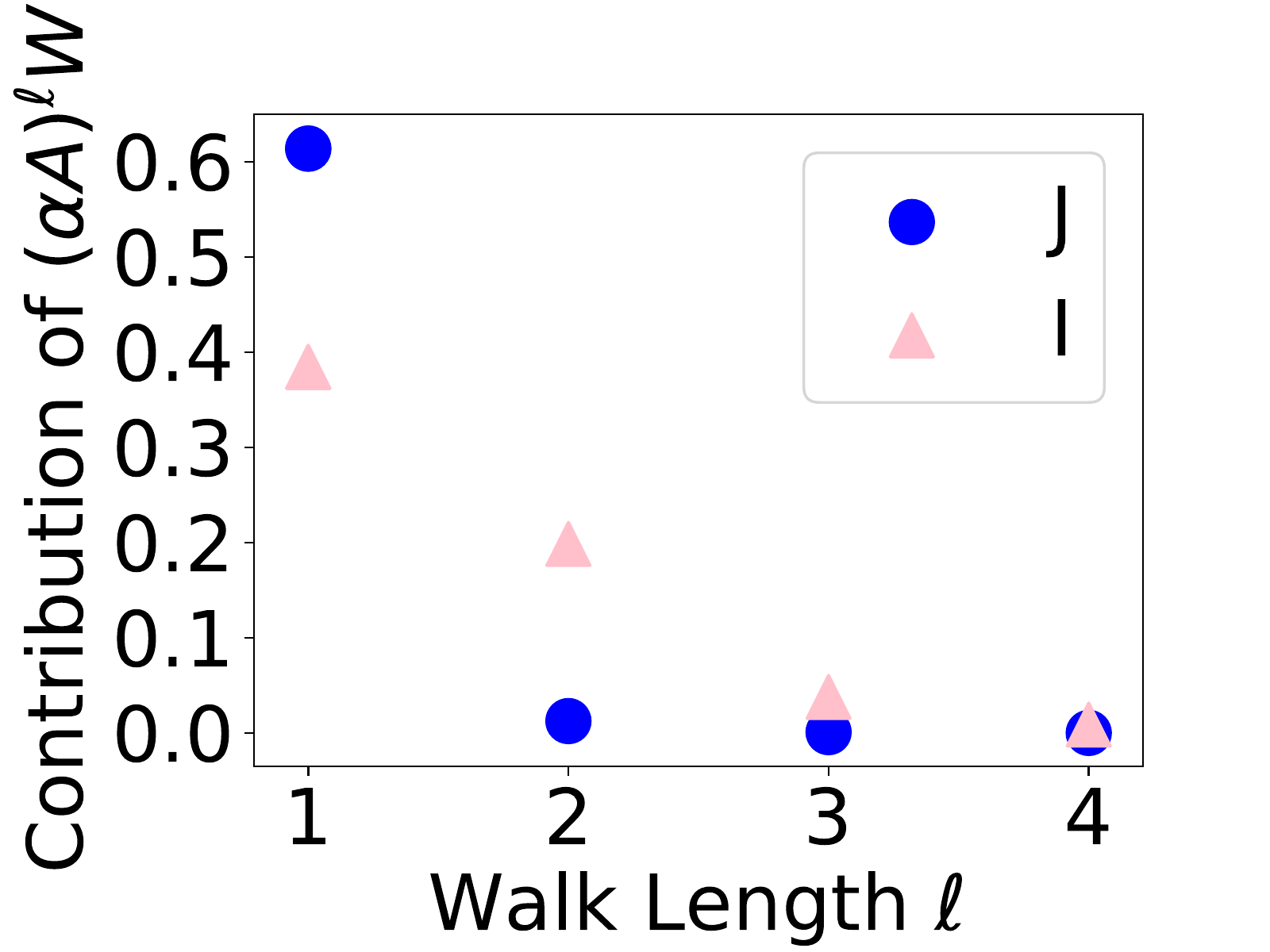}
         \caption{$\alpha = 1$}
         \label{fig:alpha1}
     \end{subfigure}
     \hfill
     \begin{subfigure}[b]{0.45\textwidth}
         \centering
         \includegraphics[width=\textwidth]{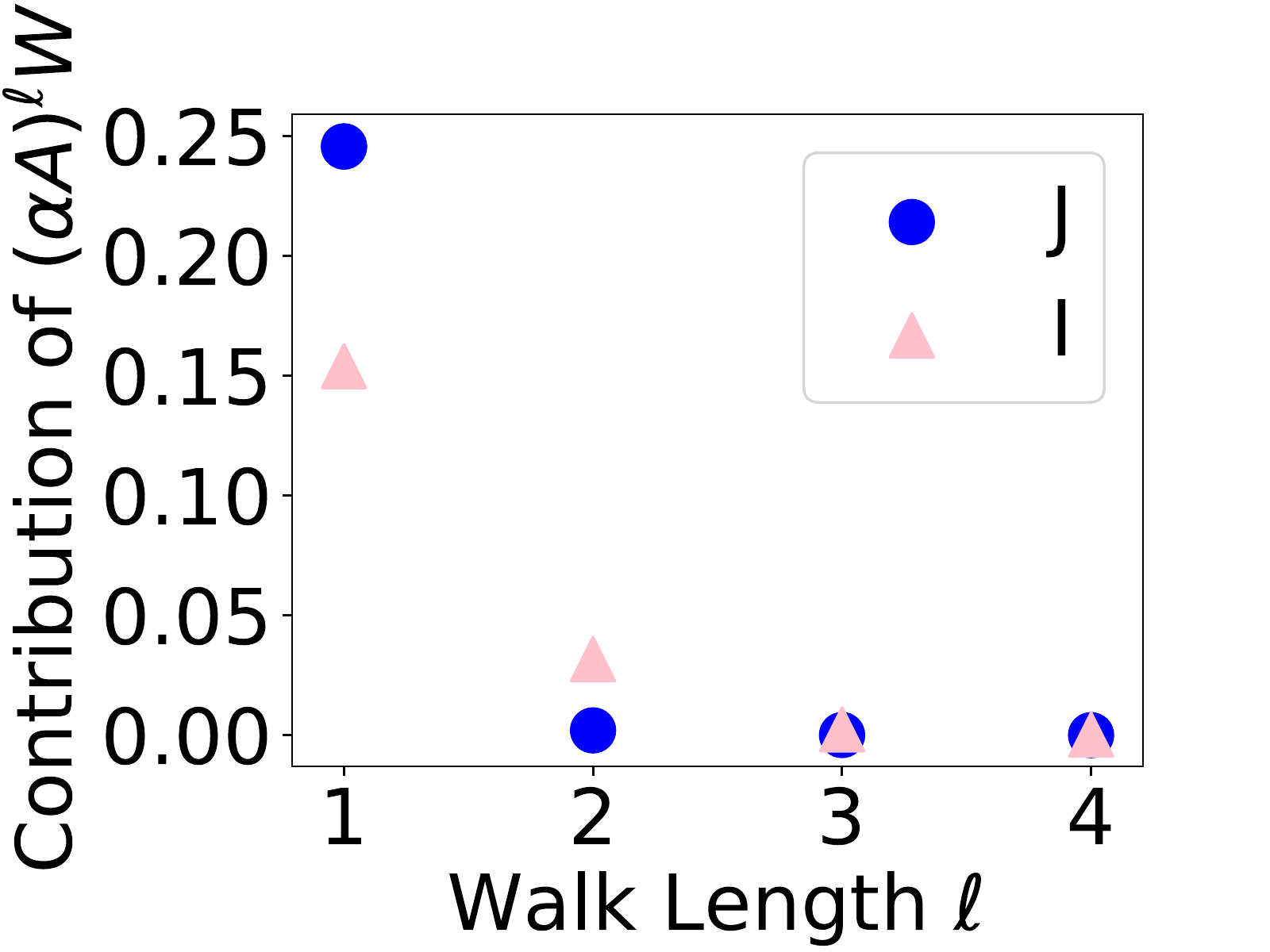}
         \caption{$\alpha = 0.4$}
         \label{fig:alpha4}
     \end{subfigure}
        \caption{Contribution of walks of length $\ell$ to the scores of nodes $I$ and $J$ in $1$-Katz centrality (a) and $0.4$-Katz centrality (b).}
        \label{fig:IJcontributions}
\end{figure}

\begin{figure}[ht]
    \begin{center}
        \begin{tikzpicture}
            \begin{scope}[every node/.style={circle,thick,draw}]
                \node (K) at (\y*0.86602540378,\y*0.5) {K};
                \node (I) at (\y*0.5,\y*0.86602540378) {I};
                \node (C) at (\y*0,\y*1) {C};
                \node (G) at (-\y*0.5,\y*0.86602540378) {G};
                \node (F) at (-\y*0.86602540378,\y*0.5) {F};
                \node (B) at (-\y*1,\y*0) {B};
                \node (A) at (\y*0.86602540378,-\y*0.5) {A};
                \node (J) at (\y*0.5,-\y*0.86602540378) {J};
                \node (E) at (\y*0,-\y*1) {E};
                \node (H) at (-\y*0.5,-\y*0.86602540378) {H};
                \node (D) at (-\y*0.86602540378,-\y*0.5) {D};
                \node (L) at (-\y*-1,-\y*0) {L};
            \end{scope}
            
            \begin{scope}[-latex]
                \path [->,-latex,line width=\x*0.22360679775mm] (A) edge node {} (J);
                \path [->,-latex,line width=\x*0.0316227766mm] (A) edge node {} (I);
                \path [->,-latex,line width=\x*0.22360679775mm] (A) edge node {} (L);
                \path [->,-latex,line width=\x*0.31622776601mm] (B) edge node {} (D);
                \path [->,-latex,line width=\x*0.31622776601mm] (B) edge node {} (I);
                \path [->,-latex,line width=\x*0.0316227766mm] (B) edge node {} (L);
                \path [->,-latex,line width=\x*0.31622776601mm] (C) edge node {} (G);
                \path [->,-latex,line width=\x*0.14142135623mm] (C) edge node {} (I);
                \path [->,-latex,line width=\x*0.0316227766mm] (D) edge node {} (L); 
                \path [->,-latex,line width=\x*0.0316227766mm] (E) edge node {} (L); 
                \path [->,-latex,line width=\x*0.0316227766mm] (G) edge node {} (F);
                \path [->,-latex,line width=\x*0.0316227766mm] (H) edge node {} (L); 
                \path [->,-latex,line width=\x*0.14142135623mm] (K) edge node {} (I); 
                \path [->,-latex,line width=\x*0.0316227766mm] (K) edge node {} (L); 
                \path [->,-latex,line width=\x*0.0316227766mm] (K) edge node {} (G); 
            \end{scope}
            \draw [-stealth, line width=\x*1.3953494186mm](A) -- (\z*0.86602540378,-\z*0.5);
            \draw [-stealth, line width=\x*0.91923881554
mm](B) -- (-\z*1,\z*0);
            \draw [-stealth, line width=\x*0.86602540378mm](C) -- (\z*0,\z*1);
            \draw [-stealth, line width=\x*0.67453687816mm](D) -- (-\z*0.86602540378,-\z*0.5);
            \draw [-stealth, line width=\x*0.62369864518mm](E) -- (\z*0,-\z*1);
            \draw [-stealth, line width=\x*0.7582875444mm](F) -- (-\z*0.86602540378,\z*0.5);
            \draw [-stealth, line width=\x*0.50990195135mm](G) -- (-\z*0.5,\z*0.86602540378);
            \draw [-stealth, line width=\x*0.65954529791mm](H) -- (-\z*0.5,-\z*0.86602540378);
            \draw [-stealth, line width=\x*1.08531101533mm](I) -- (\z*0.5,\z*0.86602540378);
            \draw [-stealth, line width=\x*0.38078865529mm](J) -- (\z*0.5,-\z*0.86602540378);
            \draw [-stealth, line width=\x*0.86429161745mm](K) -- (\z*0.86602540378,\z*0.5);
            \draw [-stealth, line width=\x*0.35496478698mm](L) -- (-\z*-1,-\z*0);
        \end{tikzpicture}
    \end{center}
    \caption{Subnetwork of the SIN with 12 most influential nodes with weighted edges. The thickness of the edges that point outwards illustrates the total out-going edge weight towards nodes in the remainder of the network.}
    \label{fig:SIN_subgraph}
\end{figure}
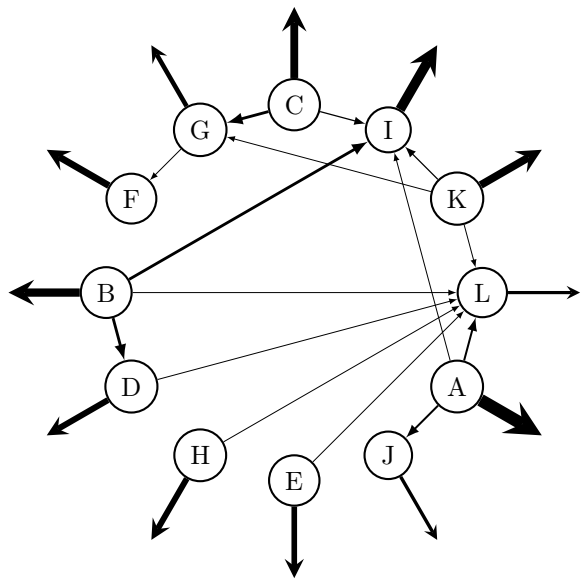
\FloatBarrier
\section{Testing the Upper Bound on Simulated Data}\label{sec:comparisons}

We would like to better understand for which values of $\ell$ that $C(\alpha)$ and $C(\alpha,\ell)$ $\epsilon$-agree. We define the \textit{true value} to be the smallest such $\ell$, and Theorem \ref{thm:epsilonagree} gives an upper bound on the true value. In Figure \ref{fig:compareupperandtrue}, we compare the upper bound to the true value on two families of undirected graphs. At each instance of $\alpha_0$ in Figures \ref{fig:compareupperandtruergg} and \ref{fig:compareupperandtruechung}, we sample 10 graphs from each specified graph family and plot the average upper bound and average true value across the samples. 

In Figure \ref{fig:compareupperandtruergg}, we sample from the Erd\H{o}s-R\'enyi model $G(n,p)$ where $n = 1000$ and $p = 0.008$. In Figure \ref{fig:compareupperandtruechung} we sample from the Chung-Lu model, which takes as input a list of node degrees. We individually sampled 1000 numbers from the negative binomial distribution to create this list. This is a distribution which takes a probability $p$ of success and a number $n$ of desired successes. We use $p = 0.1$ as the probability of success, and $n=3$ as the number of desired successes. 

The Chung-Lu model, as described here, produces graphs with a longer-tailed degree distribution than graphs sampled from the Erd\H{o}s-R\'enyi model. In Figure \ref{fig:compareupperandtrue}, the value for $\epsilon$ that we use is calculated using the same technique as in Section \ref{sec:ApplicationToSIN} for each iteration. That is, we create a list of pairwise differences between the Katz scores and set $\epsilon$ so that $5\%$ of the differences are less than it.

\begin{figure}[ht]
\centering
    \begin{subfigure}[b]{0.49\textwidth}
        \centering
        \includegraphics[scale = .36]{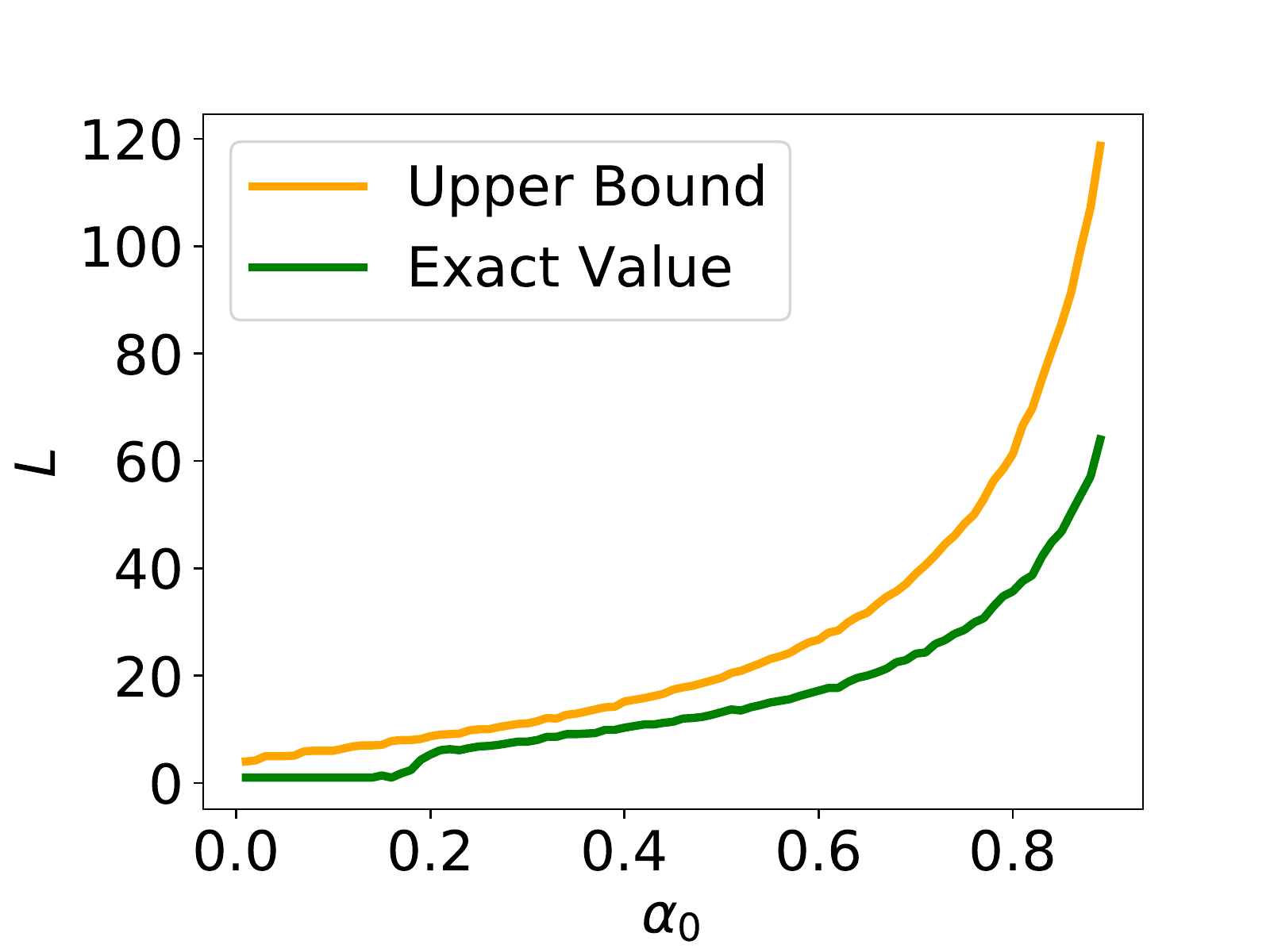}
        \caption{Erd\H{o}s-R\'enyi}
        \label{fig:compareupperandtruergg}
    \end{subfigure}
    \hfill
    \begin{subfigure}[b]{0.49\textwidth}
        \centering
        \includegraphics[scale = .36]{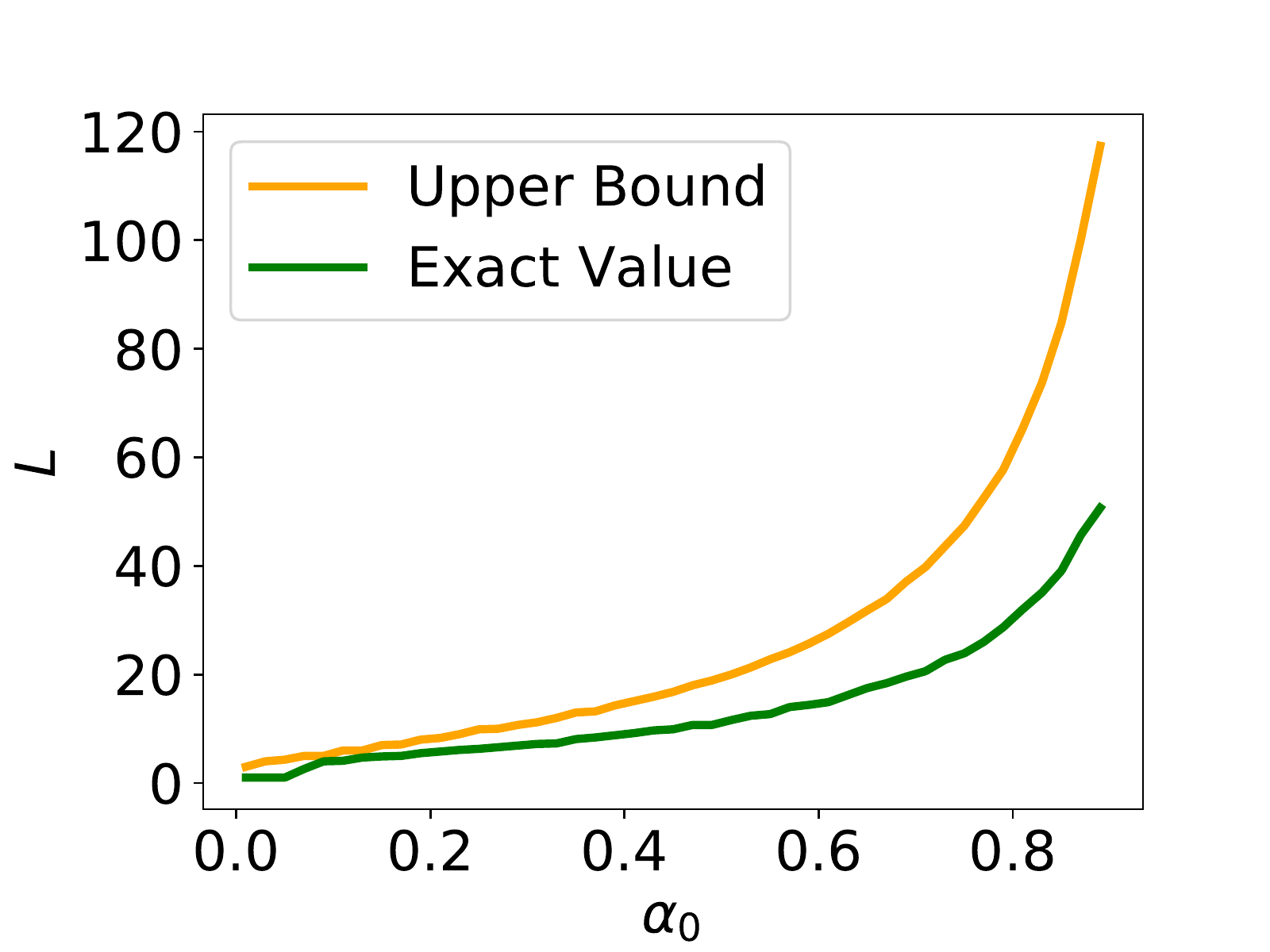}
        \caption{Chung-Lu}
        \label{fig:compareupperandtruechung}
    \end{subfigure}
    \caption{Comparing the upper bound to the true length using the Erd\H{o}s-R\'enyi model $G(1000,0.008)$ and the Chung-Lu model with degrees sampled from the negative binomial distribution.}
    \label{fig:compareupperandtrue}
\end{figure}

\section{Discussion and Conclusions}\label{sec:Discussion}
In this paper we introduced a tool to help compare different centrality measures. We applied this to $\alpha$- and $(\alpha,\ell)$-Katz centrality, to help better understand the effect of the $\alpha$-parameter on the walk lengths considered when it comes to ranking of nodes. For a given $\alpha$, we provide an upper bound on the walk length $L$ so that, when $\ell > L$, the Katz scores in $\alpha$- and $(\alpha,\ell)$-Katz centrality are within $2\epsilon$ of each other. We show that any two nodes with both centrality measures differing by at least $\epsilon$ are in the same order in both rankings when $\ell > L$. 

If we want to find a minimal value of $L$ such that $\alpha$-Katz and $(\alpha,\ell)$-Katz centrality $\epsilon$-agree for all $\ell\geq L$, then Theorem~\ref{thm:epsilonagree} provides an upper bound. We can find a stronger upper bound by iteratively increasing $L$ until for all $\ell\geq L$, $\|C(\alpha) - C(\alpha,\ell) \| < 2\epsilon$. This guarantees $\epsilon$-agreement, although $\epsilon$-agreement may still happen sooner, so care should be taken when interpreting these bounds.

All of the results in this paper require $\alpha$ to be in $(0,1/\|A\|_2)$, a subset of the possible values that can be used as a Katz parameter. It may be possible to extend these results to all possible Katz parameters, namely, values in $(0,1/\rho)$. These ranges match for undirected graphs, so this applies to directed, edge-weighted graphs.

We showed the effect of changing the $\alpha$ parameter in the Susceptibility Inference Network. In this case, changes in ranking were visible even among the top 10 nodes, and they may have major implications for decision making. It is therefore important that the choice of $\alpha$ is made carefully and tailored to each individual application. This also holds true for $\epsilon$.

\FloatBarrier
\section*{Acknowledgements}\label{sec:Acknowledgements}
The authors would like to acknowledge the Cross-Cutting Computational Modeling Project team at Glenn for their expertise, especially Dr. Drayton Munster for the helpful discussions and feedback throughout this work. The authors are grateful for NASA's Human Research Program and the University Space Research Association (USRA) for supporting this work.

\end{document}